\documentclass[runningheads]{llncs}
\usepackage{graphicx}
\usepackage{comment}
\usepackage{enumitem}
\usepackage{amsmath}
\usepackage{amsfonts}
\usepackage{mathtools}
\usepackage{color,soul}
\usepackage[mathscr]{euscript}
 \let\mathscr\relax% just so we can load this and rsfs

\newcommand{\edesirs}{\mathcal{E}}
\newcommand{\pspace}{\varOmega}

\newcommand{\domain}{\mathcal{K}}

\newcommand{\gambles}{\mathcal{L}}
\newcommand{\posi}{\mathsf{posi}}

\newcommand{\desirs}{\mathcal{D}}
\newcommand{\pr}{P}
\newcommand{\lpr}{{\underline{\pr}}}

\newcommand{\partit}{\mathcal{P}}

\begin{document}
\title{Algebras of Sets and Coherent Sets of Gambles}
\author{Juerg Kohlas\inst{1}\orcidID{0000-0001-8427-5594}\and
Arianna Casanova\inst{2}\orcidID{0000-0001-6765-0292} \and
Marco Zaffalon\inst{2}\orcidID{0000-0001-8908-1502}}
\authorrunning{Kohlas et al.}
\institute{Department of Informatics DIUF, University of Fribourg, Switzerland \\
\email{juerg.kohlas@unifr.ch}\\
\and
Istituto Dalle Molle di Studi sull’Intelligenza Artificiale (IDSIA), Switzerland\\
\email{\{arianna,zaffalon\}@idsia.ch}}
\maketitle              % typeset the header of the contribution
\begin{abstract}
%The abstract should briefly summarize the contents of the paper in
%15--250 words.
In a recent work we have shown how to construct an information algebra of coherent sets of gambles defined on general possibility spaces. Here we analyze the connection of such an algebra with the \emph{set algebra} of subsets of the possibility space on which gambles are defined and the \emph{set algebra} of sets of its \emph{atoms}. Set algebras are particularly important information algebras since they are their \emph{prototypical} structures.  Furthermore, they are the algebraic counterparts of classical propositional logic. As a consequence, this paper also details how propositional logic is naturally embedded into the theory of \emph{imprecise probabilities}. %As before, also in this work we consider general possibility spaces.

\end{abstract}

\section{Introduction and Overview}
While analysing the compatibility problem of coherent sets of gambles, Miranda and Zaffalon \cite{mirzaffalon20} have recently remarked that their main results could be obtained also using the theory of information algebras \cite{kohlas03}.
%This point of view however, was not worked out in \cite{mirzaffalon20}.  
This observation has been taken up and deepened in some of our recent work \cite{kohlas21,kohlas21b}: we have shown that the founding properties of desirability can in fact be abstracted into properties of information algebras. Stated differently, desirability makes up an information algebra of coherent set of gambles. %A similar scope was pursued by De Cooman in \cite{decooman05}. He discovered indeed, that there is a common order-theoretic structure underlying many of the models for representing beliefs in the literature, including lower previsions and sets of almost desirable gambles. Even if they share some elements, the latter focuses more on the study of \emph{belief dynamics} (belief expansion and belief revision). An explicit comparison of the two approaches nevertheless, can be found in our previous work \cite{kohlas21}.

Information algebras are algebraic structures composed by `pieces of information' that can be manipulated by operations of \emph{combination}, to aggregate them, and \emph{extraction}, to extract information regarding a specific question.
From the point of view of information algebras, sets of gambles defined on a possibility space $\pspace$ are pieces of information about $\pspace$. %In particular, this paper is intended as an extension of our previous paper, in which we treat the particular case where
%information one is interested in concerns the values of certain groups of
%variables $\{X_i: \; i \in I\}$ with $I$ an index set, $\pspace = \bigtimes_{i \in I} \pspace_i$, where $\pspace_i$ is the set of possible values of $X_i$, and $\omega \equiv_S \omega' \iff \omega|_S=  \omega'|_S$, \footnote{If we think of $\omega \in \pspace$, as a map $\omega: I \rightarrow \pspace$, $\omega|_S$ is the restriction of the map $\omega$ to $S$.} for every $S \subseteq I$ and $\omega, \omega' \in \pspace$. (see \cite{kohlas21}). %to the case in which the possibility space on which gambles are defined,  does not necessarily factorize in 
%domains of variables. 
  %Such pieces of information can be aggregated and the information they contain about specific questions can be extracted.  This leads to an algebraic structure satisfying a number of simple axioms.  
 It is well known that coherent sets of gambles are ordered by inclusion and, in this order, there are maximal elements~\cite{CooQua12}. In the language of information algebras such elements are called \emph{atoms}. In particular, any coherent set of gambles is contained in a maximal set (an atom) and it is the intersection (meet) of all the atoms it is contained in. An information algebra with these properties is called atomistic.
 Atomistic information algebras have the universal property of being embedded in a set algebra, which is an information algebra whose elements are sets. This is an important representation theorem for information algebras, since set algebras are a special kind of algebras based on the usual set operations. Conversely, any such set algebra of subsets of $\pspace$ is embedded in the algebra of coherent sets of gambles defined on $\pspace$. %This has some importance for conditioning ***why?***, a subject which however is not pursued in this paper. 
 These links between set algebras and the algebra of coherent sets of gambles are the main topic of the present work.

 After recalling the main concepts introduced in our previous work in Sections~\ref{sec:DesGambles}--\ref{sec:InfAlgs}, in Section~\ref{sec:Atoms} we establish the basis to show that sets of atoms of the information algebra of coherent sets of gambles, form indeed a set algebra. In Section \ref{sec:homomorphism} we define the concept of \emph{embedding} for information algebras and finally, in Section \ref{sec:SetAlg}, we show the links between set algebras (of subsets of $\pspace$ and of sets of atoms) and the algebra of coherent sets of gambles. %Finally, in Section \ref{sec:SetAlg}, we show that an instance of a very important class of information algebras, called \emph{set algebras}, can be embedded into the information algebra of coherent sets of gambles. This is an interesting result because set algebras can be considered as the \emph{prototypes} of information algebras (see \cite{kohlas03}).
 
 Since set algebras are algebraic counterparts of classical propositional logic, the results of this paper details how the latter is formally part of the theory of imprecise probabilities \cite{walley91}. We refer also to \cite{decooman05} for another aspect of this issue.

\section{Desirability} \label{sec:DesGambles}

Consider a set $\pspace$ of possible worlds. A gamble over this set is a bounded function
$f : \pspace \rightarrow \mathbb{R}$.
It is interpreted as an uncertain reward in a linear utility scale. A subject might desire a gamble or not, depending on the information they have about the experiment whose possible outcomes are the elements of $\pspace$.
We denote the set of all gambles on $\pspace$ by $\mathcal{L}(\pspace)$, or more simply by $\mathcal{L}$, when there is no possible ambiguity. We also introduce $\mathcal{L}^+(\pspace) \coloneqq \{ f \in \mathcal{L}(\pspace): \; f\geq 0, f \not= 0\}$, or simply $\mathcal{L}^+$ when no ambiguity is possible, the set of non-negative non-vanishing gambles. These gambles should always be desired, since they may increase the wealth with no risk of decreasing it.
As a consequence of the linearity of our utility scale, we assume also that a subject disposed to accept the transactions represented by $f$ and $g$, is disposed to accept also $\lambda f + \mu g$ with $\lambda, \mu \ge 0$ not both equal to $0$.
More generally speaking, we consider the notion of a coherent set of gambles \cite{walley91}:
%A coherent set of gambles over $\pspace$ is a subset $\desirs$ of $\mathcal{L}(\pspace)$ such that
\begin{comment}
\begin{enumerate}
\item $\mathcal{L}^+ \subseteq \desirs$,
\item $0 \not\in \desirs$,
\item $f,g \in \desirs$ implies $f + g \in \desirs$,
\item $f \in \desirs$, and $\lambda > 0$ implies $\lambda \cdot f \in \desirs$.
\end{enumerate}
\end{comment}
\begin{definition}[\textbf{Coherent set of gambles}]
We say that a subset $\desirs$ of $\gambles$ is a \emph{coherent} set of gambles if and only if $\desirs$ satisfies the following properties:
\begin{enumerate}[label=\upshape D\arabic*.,ref=\upshape D\arabic*]
\item\label{D1} $\gambles^+ \subseteq \desirs$ [Accepting Partial Gains];
\item\label{D2} $0\notin \desirs$ [Avoiding Status Quo];
\item\label{D3} $f,g \in \desirs \Rightarrow f+g \in \desirs$ [Additivity];
\item\label{D4} $f \in \desirs, \lambda>0 \Rightarrow \lambda f \in \desirs$ [Positive Homogeneity].
\end{enumerate}
\end{definition}
So, $\desirs$ is a convex cone. Let us denote with $C(\pspace)$, or simply with $C$, the family of coherent sets of gambles on $\pspace$.
This leads to the concept of natural extension.
\begin{definition}[\bf{Natural extension for gambles}] \label{def:natex}Given a set $\domain\subseteq\gambles$, we call $
\edesirs(\domain) \coloneqq\posi(\domain\cup\gambles^+)$,
where $\posi(\domain')\coloneqq\left\{ \sum_{j=1}^{r} \lambda_{j}f_{j}: f_{j} \in \domain', \lambda_{j} > 0, r \ge 1\right\}$,
for every set $\domain' \subseteq \gambles$, its \emph{natural extension}.
\end{definition}
%If $\desirs'$ is any subset of $\mathcal{L}(\pspace)$, then 
%\begin{eqnarray*}
%\mathcal{E}(\desirs') = \posi(\mathcal{L}^+(\pspace)\cup \desirs'),
%\end{eqnarray*}
%is called the natural extension of a set of gambles, where $\posi(\desirs)$ denotes all finite positive linear combinations $\lambda_1 f_1 + \ldots + \lambda_n f_n$. $\lambda_i > 0$ of elements of $\desirs$. 
The natural extension $\edesirs(\desirs)$ of a set of gambles $\desirs$ is coherent if and only if $0 \not\in \edesirs(\desirs)$.

%Coherent sets are closed under intersection, that is they form a topless $\cap$-structure. By standard order theory (see \cite{daveypriestley97}), they are ordered by inclusion, intersection is meet in this order and a join exists if they have an upper bound among coherent sets
%\begin{eqnarray*}
%\bigvee_{i \in I} \desirs_i \coloneqq \bigcap \{\desirs \in C(\pspace):  \bigcup_{i \in I} \desirs_i \subseteq \desirs \},
%\end{eqnarray*}
%if we indicate with $C(\pspace)$, or simply with $C$, the family of coherent sets of gambles on $\pspace$.

%Notice also that, if $0 \notin \mathcal{E}(\desirs')$, it is the smallest coherent set containing $\desirs'$. Therefore, if $\mathcal{E}(\bigcup_{i \in I} \desirs_i)$ is coherent, we have
%\begin{eqnarray*}
%\mathcal{E}(\desirs') = \bigcap \{\desirs \in C(\pspace): \desirs' \subseteq \desirs\},
%\end{eqnarray*}
%so that 
%\begin{eqnarray*}
%\bigvee_{i \in I} \desirs_i = \mathcal{E}(\bigcup_{i \in I} \desirs_i).
%\end{eqnarray*}
%In view of the following section, it is convenient to add $\mathcal{L}(\pspace)$ to $C(\pspace)$ and let $\Phi(\pspace) \coloneqq C(\pspace) \cup \{\mathcal{L}(\pspace)\}$. The family of sets in $\Phi(\pspace)$, or simply $\Phi$ where there is no possible ambiguity, is still a $\cap$-structure, but now a topped one. 

In \cite{kohlas21b} we showed that $\Phi(\pspace) \coloneqq C(\pspace) \cup \{\mathcal{L}(\pspace)\}$, or simply $\Phi$ when there is no possible ambiguity, is a complete lattice under inclusion \cite{daveypriestley97}, meet is intersection and join is defined for any family of sets $\desirs_i \in \Phi$ as
\begin{equation*}
\bigvee_{i \in I} \desirs_i \coloneqq \bigcap \left\{\desirs \in \Phi: \bigcup_{i \in I} \desirs_i \subseteq \desirs\right\}.
\end{equation*}
Note that, if the family of coherent sets $\desirs_i$ has no upper bound in $C$, then its join is simply $\mathcal{L}$. Moreover, we defined the following closure operator \cite{daveypriestley97} on subsets of gambles.
\begin{equation}\label{eq:closureoperatorC}
\mathcal{C}(\desirs') \coloneqq \bigcap \{\desirs \in \Phi: \desirs' \subseteq \desirs\}.
\end{equation}
%It can be shown that is a closure (or consequence) operator on subsets of gambles (see \cite{daveypriestley97}).
%\begin{lemma}
%The operator defined in \eqref{eq:closureoperatorC} is a closure operator on subsets of $\gambles$, i.e. for every $\desirs,\; \desirs' \subseteq \gambles$:
%\begin{enumerate}
%\item $\desirs \subseteq \mathcal{C}(\desirs)$;
%\item $\desirs \subseteq \desirs'$ implies $\mathcal{C}(\desirs) \subseteq \mathcal{C}(\desirs')$;
%\item $\mathcal{C}(\mathcal{C}(\desirs)) = \mathcal{C}(\desirs)$.
%\end{enumerate}
%\end{lemma}
%For further reference, it is easy to prove also the following well-known result.
%\begin{lemma} \label{le:UnionClosSets}
%For any $\desirs_1, \desirs_2 \subseteq \gambles$ we have:
%\begin{eqnarray*}
%\mathcal{C}(\mathcal{C}(\desirs_1) \cup \desirs_2) = \mathcal{C}(\desirs_1 \cup \desirs_2).
%\end{eqnarray*}
%\end{lemma}
It is possible to notice that $\mathcal{C}(\desirs) = \mathcal{E}(\desirs)$ if $0 \not\in \mathcal{E}(\desirs)$, that is if $\mathcal{E}(\desirs)$ is coherent. Otherwise we may have $\mathcal{E}(\desirs) \not= \mathcal{L}(\pspace)$.  %These results prepare the way to an information algebra of coherent sets of gambles (see Section \ref{sec:InfAlgs}). 

The most informative cases of coherent sets of gambles, i.e., coherent sets that are not proper subsets of other coherent sets, are called \textit{maximal}. The following proposition provides a characterisation of such maximal elements~\cite[Proposition~2]{CooQua12}.
\begin{proposition}[\textbf{Maximal coherent set of gambles}]
A coherent set of gambles $\desirs$ is \emph{maximal} if and only if
\begin{equation*}
(\forall f \in \gambles \setminus \{0\})\ f \notin \desirs \Rightarrow -f \in \desirs.
\end{equation*}
\end{proposition}
We shall denote maximal sets with $M$ to differentiate them from the general case of coherent sets.
These sets play an important role with respect to information algebras (see Section \ref{sec:Atoms}).
Another important class is that of \emph{strictly desirable} sets of gambles \cite{walley91}.\footnote{Strictly desirable sets of gambles are important because they are in a one-to-one relation with \emph{coherent lower previsions}; these are a generalization of the usual expectation operator on gambles. Given a coherent lower prevision $\lpr(\cdot)$, $D^+ \coloneqq \{ f \in \gambles: \lpr(f) >0 \} \cup \gambles^+$ is a strictly desirable set of gambles~\cite[Section 3.8.1]{walley91}.}
\begin{definition} [\textbf{Strictly desirable set of gambles}]
A coherent set of gambles $D$ is said to be \emph{strictly desirable} if and only if it satisfies
$$(\forall f  \in D \setminus \gambles^+)(\exists \delta >0)\ f- \delta \in D.$$
\end{definition}
For strictly desirable sets, we shall employ the notation $D^+$. 

\section{Stucture of Questions and Possibilities} \label{sec:Questions}

In this section we review the main results about the structure of $\pspace$ \cite{kohlas17,kohlasmonney95,kohlas21b}.
With reference to our previous work \cite{kohlas21b}, we recall that coherent sets of gambles are understood as pieces of information describing beliefs about the elements in $\pspace$. Beliefs may be originally expressed relative to different questions or variables that we identify by families of equivalence relations $\equiv_x$ on $\pspace$ for $x$ in some index set $Q$.
A question $x \in Q$ has the same answer in possible worlds $\omega \in \pspace$ and $\omega' \in \pspace$, if $\omega \equiv_x \omega'$.

%It is possible to consider coherent sets of gambles in a more general frame. 
%As before, let $\pspace$ be a set of possible worlds. We consider families of equivalence relations $\equiv_x$ for $x$ in some index set $Q$. Informally, we mean that $Q$ represents questions, and a question $x \in Q$ has the same answer in possible worlds $\omega$ and $\omega'$, if $\omega \equiv_x \omega'$. 
There is a partial order between questions capturing granularity: question $y$ is finer than question $x$ if $\omega \equiv_y \omega'$ implies $\omega \equiv_x \omega'$. This can be expressed also considering partitions $\partit_x, \; \partit_y$ of $\pspace$ whose blocks are respectively, the equivalence classes $[\omega]_x, \; [\omega]_y$ of the equivalence relations $\equiv_x, \; \equiv_y$, representing possible answers to $x$ and $y$.
Then $\omega \equiv_y \omega'$ implies $\omega \equiv_x \omega'$, meaning that any block $[\omega]_y$ of partition $\partit_y$ is contained in some block $[\omega]_x$ of partition $\partit_x$. %The partition of $\pspace$ by $\partit_x$ is coarser than the one by $\partit_y$. 
If this is the case, we say equivalently that: $x \le y$ or $\partit_x \leq \partit_y$.\footnote{In the literature usually the inverse order between partitions is considered. However, this order better corresponds to our natural order of questions by granularity.}
%if any block of $\partit_y$ is contained in a block of $\partit_x$ \footnote{In order literature usually the inverse order between partitions is considered. However, this order corresponds better to our natural order of question by granularity.}. This order captures the fineness or granularity of questions.
Partitions $Part(\pspace)$ of any set $\pspace$, form a lattice under this order \cite{graetzer03}. In particular, the partition $\sup\{\partit_x,\partit_y\} \coloneqq \partit_x \vee \partit_y$ of two partitions $\partit_x, \partit_y$ is, in this order, the partition obtained as the non-empty intersections of blocks of $\partit_x$ with blocks of $\partit_y$. It can be equivalently expressed also as $\partit_{x \vee y}$. Definition of meet $\partit_x \wedge \partit_y$, or equivalently $\partit_{x \wedge y}$, is somewhat involved \cite{graetzer03}. %When partitions \emph{commute}, it has a simpler formulation (see \cite{kohlas21b}).
We usually assume that the set of questions $Q$ analyzed, considered together with their associated partitions denoted with $\mathcal{Q} \coloneqq \{\partit_x:x \in Q\}$, is a join-sub-semilattice of $(Part(\pspace),\leq)$ \cite{daveypriestley97}. In particular, we assume often that the top partition in $Part(\pspace)$, i.e. $\partit_\top$ (where the blocks are singleton sets $\{\omega\}$ for $\omega \in \pspace$), belongs to $\mathcal{Q}$.
%The join $\partit_x \vee \partit_y$ corresponds to the combined questions $x$ and $y$, i.e $\omega \equiv_{x \vee y} \omega'$ if and only if $\omega \equiv_x \omega'$ and $\omega \equiv_y \omega'$. 
%We may transport the order between partitions to $Q$ and viceversa: $x \leq y$ iff $\partit_x \leq \partit_y$ and we have then $\partit_x \vee \partit_y = \partit_{x \vee y}$ and $\partit_x \wedge \partit_y = \partit_{x \wedge y}$. Furthermore, we assume also often that the top partition in $Part(\pspace)$, i.e. $\partit_\top$ (where the blocks are singleton sets $\{\omega\}$ for $\omega \in \pspace$) belongs to $\mathcal{Q}$.
A gamble $f$ on $\pspace$ is called \emph{$x$-measurable}, iff for all $\omega \equiv_x \omega'$ we have $f(\omega) = f(\omega')$, that is, if $f$ is constant on every block of $\partit_x$. It could then also be considered as a function (a gamble) on the set of blocks of $\partit_x$. We denote with $\mathcal{L}_x(\pspace)$, or more simply with $\mathcal{L}_x$ when no ambiguity is possible, the set of all $x$-measurable gambles. %In particular we have $x \leq y$ if and only if $\mathcal{L}_x$ is a subspace of $\mathcal{L}_y$ and $\mathcal{L}_x,\mathcal{L}_y \subseteq \mathcal{L}_{x \vee y}$. %Note that $\mathcal{L}_x$ as well as $\mathcal{L}(\pspace)$ is a linear space for all $x$.
 %$x \leq y$ if and only if $\mathcal{L}_x$ is a subspace of $\mathcal{L}_y$. So we have $\mathcal{L}_x,\mathcal{L}_y \subseteq \mathcal{L}_{x \vee y}$. In fact $\mathcal{L}_{x \vee y}$ is the smallest subspace containing $\mathcal{L}_x$ and $\mathcal{L}_y$.
%This is so, because if $\mathcal{L}_x,\mathcal{L}_y \subseteq \mathcal{L}_z$, then $x,y, \leq z$, hence $x \vee y \leq z$.

%Sometimes we want to consider partitions so that $\mathcal{L}_{x \wedge y} = \mathcal{L}_x \cap \mathcal{L}_y$. We show below (Lemma \ref{le:CommRel}) that for this case it is sufficient and necessary that $\omega \equiv_{x \wedge y} \omega'$ implies the existence of an $\omega''$ so that $\omega \equiv_x \omega'' \equiv_y \omega'$.

%We consider below coherent sets of gambles as pieces of information, describing beliefs about the likeliness of the possibilities in $\pspace$. However, we may be interested in the content of this information relative to some question $x \in Q$, and we propose how to extract this part of information from the original one. Also, possible beliefs may be originally expressed relative to different questions, and these pieces of information must be combined to an aggregated belief. 

%This leads then to an algebraic structure, called an information algebra (see \cite{kohlas03}). In the form sketched here, it will be, more precisely, a \emph{domain-free information algebra} (Section \ref{sec:InfAlgs}). %Later on, in Section \ref{sec:LabInfAlg}, we consider a \emph{labeled} version of the algebra.
We recall also the logical independence and conditional logical indipendence relation between partitions \cite{kohlas17,kohlasmonney95}. 
\begin{definition}[Independent Partitions]
For a finite set of partitions $\partit_1,\ldots,\partit_n \in Part(\pspace)$, $n \geq 2$, let us define
\begin{equation*}
R(\partit_1,\ldots,\partit_n) \coloneqq \{(B_1,\ldots,B_n):B_i \in \partit_i,\cap_{i=1}^n B_i \not= \emptyset\}.
\end{equation*}
We call the partitions \emph{independent}, if
$R(\partit_1,\ldots,\partit_n) = \partit_1 \times \cdots \times \partit_n$.
\end{definition}

%The intuition behind this definition is the following: $R( \partit_1, \dots, \partit_n)$ contains the tuples of mutually compatible blocks of $\partit_1, \dots, \partit_n$, representing compatible answers to the $n$ questions modelled by the partitions. 
%If they are independent, the answer to a question $\partit_i$ does not constrain the answers to the other questions, or in other words, it contains no information relative to the other questions. 

%Analogously, it can be also introduced a logical conditional independence relation between partitions.
\begin{definition}[Conditionally Independent Partitions]
Consider a finite set of partitions $\partit_1,\dots \partit_n \in Part(\pspace)$, and a block $B$ of a partition $\partit$ (contained or not in the list $\partit_1,\ldots,\partit_n$), then define for $n \geq 1$,
\begin{equation*}
R_B(\partit_1,\ldots,\partit_n) \coloneqq \{(B_1,\ldots,B_n):B_i \in \partit_i,\cap_{i=1}^n B_i \cap B \not= \emptyset\}.
\end{equation*}
We call $\partit_1,\ldots,\partit_n$ \emph{conditionally independent} given $\partit$, iff for all blocks $B$ of $\partit$,
$R_B(\partit_1,\ldots,\partit_n) = R_B(\partit_1) \times \cdots \times R_B(\partit_n)$.
\end{definition}
%So, $\partit_1,\dots,\partit_n$ are conditionally independent given $\partit$, if knowing an answer to $\partit_i$ compatible with $B \in \partit$, gives no information on the answers to the other questions, except that they must each be compatible with $B$. 
This relation holds if and only if $B_i \cap B \not= \emptyset$ for all $i=1,\ldots,n$, implies that $B_1 \cap \ldots \cap B_n \cap B \not= \emptyset$. In this case we write $\bot\{\partit_1,\ldots,\partit_n\} \vert \partit$ or, for $n =2$, $\partit_1 \bot \partit_2 \vert \partit$. $\partit_x \bot \partit_y \vert \partit_z$ can be indicated also with $x \bot y \vert z$. We may also say that  $\partit_1 \bot \partit_2 \vert \partit$ if and only if $\omega \equiv_{\partit} \omega'$ implies the existence of an element $\omega'' \in \pspace$ such that $\omega \equiv_{\partit_1 \vee \partit} \omega''$ and $\omega' \equiv_{\partit_2 \vee \partit} \omega''$.
The three-place relation $\partit_1 \bot \partit_2 \vert \partit$ is, in particular, a \emph{quasi-separoid} \cite{kohlas17}, a retract of the concept of separoid \cite{daveypriestley97}.
\begin{theorem} \label{th:QSepOfPart}
Given $\partit,\partit',\partit_1,\partit_2 \in Part(\pspace)$, we have:
\
\begin{description}
\item[C1] $\partit_1 \bot \partit_2 \vert \partit_2$;
\item[C2] $\partit_1 \bot \partit_2 \vert \partit$ implies $\partit_2 \bot \partit_1 \vert \partit$;
\item[C3] $\partit_1 \bot \partit_2 \vert \partit$ and $\partit' \leq \partit_2$ imply $\partit_1 \bot \partit' \vert \partit$;
\item[C4] $\partit_1 \bot \partit_2 \vert \partit$ implies $\partit_1 \bot \partit_2 \vee \partit \vert \partit$.
\end{description}
\end{theorem}
%A three-place relation like $\partit_1 \bot \partit_2 \vert \partit$ satisfying C1 to C4 has been called a \textit{quasi-separoid} (q-separoid) in \cite{kohlas17}. It is a reduct of a separoid, a concept discussed in \cite{dawid01}, in relation to the concept of (logical) conditional independence in general.
From these properties, it follows that $\partit_x \perp \partit_y \vert \partit_z \iff \partit_{x \vee z} \perp \partit_{y \vee z} \vert \partit_z$,
which we use very often later on.

\section{Information Algebra of Coherent Sets of Gambles} \label{sec:InfAlgs}
In \cite{kohlas21b} we showed that $\Phi$ with the following operations: %Combination is meant to aggregate two or more pieces of information. Extraction filters out the part of a piece of information which is relevant to a given question $x$. 
\begin{enumerate}
\item Combination: $\desirs_1 \cdot \desirs_2 \coloneqq \mathcal{C}(\desirs_1 \cup \desirs_2)$,
\item Extraction: $\epsilon_x(\desirs) \coloneqq \mathcal{C}(\desirs \cap \mathcal{L}_x)$ for $x \in Q$,
\end{enumerate}
is a \emph{domain-free information algebra} that we call the \emph{domain-free information algebra of coherent sets of gambles}. 
Combination captures aggregation of pieces of belief, and extraction describes filtering the part of information relative to a question $x \in Q$. 
Information algebras are particular \emph{valuation algebras} as defined by \cite{shafershenoy90} but with idempotent combination. Domain-free versions of valuation algebras have been proposed by Shafer \cite{shafer91}. Idempotency of combination has important consequences, such as the possibility to define an information order, atoms, approximation, and more \cite{kohlas03,kohlas17}. It also offers---the subject of the present paper---important connections to set algebras.

Here we remind the characterizing properties of the domain-free information algebra $\Phi$ together with a system of questions $Q$ and a family $E$ of extraction operators $\epsilon_x : \Phi \rightarrow \Phi$ for $x \in Q$:
\begin{enumerate}
\item \textit{Semigroup:} $(\Phi,\cdot)$ is a commutative semigroup with a null element $0 = \gambles(\pspace)$ and a unit $1 = \gambles^+(\pspace)$.
\item \textit{Quasi-Separoid:} $(Q,\leq)$ is a join semilattice and $x \bot y \vert z$ with $x,y,z \in Q$, a quasi-separoid.
%\item \textit{Unit and Null:} For all $x \in Q$ we have $\epsilon_x(1) = 1$ and $\epsilon_x(D) = 0$ implies $D = 0$.
\item \textit{Existential Quantifier:} For any $x \in Q$, $\desirs_1,\desirs_2,\desirs \in \Phi$:
\begin{enumerate}
\item $\epsilon_x(0) = 0$,
\item $\epsilon_x(\desirs) \cdot \desirs = \desirs$,
\item $\epsilon_x(\epsilon_x(\desirs_1) \cdot \desirs_2) = \epsilon_x(\desirs_1) \cdot \epsilon_x(\desirs_2)$.
\end{enumerate}
\item \textit{Extraction:} For any $x,y,z \in Q$, $\desirs \in \Phi$, such that $x \vee z \bot y \vee z \vert z$ and $\epsilon_x(\desirs) = \desirs$, we have:
\begin{equation*}
\epsilon_{y \vee z}(\desirs) = \epsilon_{y \vee z}(\epsilon_z(\desirs)).
\end{equation*}
\item \textit{Support:} For any $\desirs \in \Phi$ there is an $x \in Q$ so that $\epsilon_x(\desirs) = \desirs$, i.e. a \emph{support} of $\desirs$ \cite{kohlas21b}, and for all $y \geq x, \; y \in Q$, $\epsilon_y(\desirs) = \desirs$.
\end{enumerate}
When we need to specify all the constructing elements of the domain-free information algebra $\Phi$, we can refer to it with the tuple $(\Phi, \mathcal{Q}, \le, \bot, \cdot, E)$, where $E$ is the family of the extraction operators constructed starting from $x \in Q$ or, equivalently, from partitions in $\mathcal{Q}$.\footnote{When we need to be explicit about partitions, we can indicate the extraction operator $\epsilon_x$ also as $\epsilon_{\mathcal{P}_x}$, where $\mathcal{P}_x \in \mathcal{Q}$ is the partition associated to the question $x \in Q$.} When we do not need this degree of accuracy, we can refer to it simply as $\Phi$. Analogous considerations can be made for other domain-free information algebras.
Notice that, in particular, $(\Phi, \cdot)$ is an idempotent, commutative semigroup. So, a partial order is defined by $\desirs_1 \leq \desirs_2$ if $\desirs_1 \cdot \desirs_2 = \desirs_2$. Then $\desirs_1 \leq \desirs_2$ if and only if $\desirs_1 \subseteq \desirs_2$. This order is called an \textit{information  order} \cite{kohlas21b}.
This definition entails the following facts: $\epsilon_x(\desirs) \leq \desirs$ for every $\desirs \in \Phi, \; x \in Q$; given $\desirs_1, \desirs_2 \in \Phi$, if $\desirs_1 \leq \desirs_2$, then $\epsilon_x(\desirs_1) \leq \epsilon_x(\desirs_2)$ for every $x \in Q$ \cite{kohlas03}.

\section{Atoms and Maximal Coherent Sets of Gambles} \label{sec:Atoms}
%Here we come back to maximal coherent sets of gambles and identify them as \textit{atoms} in the information algebra of coherent sets of gambles. An atom in an information algebra is a maximally informative set, in the information order, different from $0$. In terms of coherent sets of gambles, an atom $M$ is a an element of $\Phi$ so that for all $d \in \Pi$, we have that $M \subseteq D$ implies either $M = D$ or $D = 0 (=\mathcal{L}(\pspace)$. So, obviously, 
Maximal coherent sets $M$ are \emph{atoms} in the information algebra of coherent sets of gambles \cite{kohlas21b}. This is a well-known concept in (domain-free) information algebras. We remind the following elementary properties of atoms \cite{kohlas03}, immediately derivable from the definition. If $M,M_1$ and $M_2$ are atoms of $\Phi$ and $\desirs \in \Phi$, then
\begin{enumerate}
\item $M \cdot \desirs = M$ or $M \cdot \desirs = 0$,
\item either $\desirs \leq M$ or $M \cdot \desirs = 0$,
\item either $M_1 = M_2$ or $M_1 \cdot M_2 = 0$.
\end{enumerate}
We indicate with $At(\Phi)$ the set of all atoms of $\Phi$, and with  $At(\desirs)$ the set of all atoms $M$ which dominate $\desirs \in \Phi$, that is, $At(\desirs) \coloneqq \{M \in At(\Phi):\desirs \subseteq M\}$. %In $At(\desirs)$ there are all atoms which imply $\desirs$, or which are more informative than $\desirs$. 
Furthermore, $\Phi$ is \textit{atomic} \cite{kohlas03}, i.e. for any $\desirs \neq 0$ the set $At(\desirs)$ is not empty, and \emph{atomistic}, i.e. for any $\desirs \not= 0$, $\desirs = \bigcap At(\desirs)$.
%\begin{equation*}
%\desirs = \bigcap At(\desirs).
%\end{equation*}
%So it is also atomic composed \cite{kohlas03} or\textit {atomistic} \cite{kohlasschmid14,kohlasschmid16}. Finally, for any subset $A$ of atoms, we have also that $\bigcap A$ is a coherent set of gambles $\desirs$ such that $A \subseteq At(\desirs)$.  
It is a general result of atomistic information algebras that the subalgebras $\epsilon_x(\Phi)$ are also atomistic \cite{kohlasschmid16}. Moreover, in \cite{kohlas21b}, we showed that $At(\epsilon_x(\Phi)) = \epsilon_x(At(\Phi)) = \{\epsilon_x(M):M \in At(\Phi)\}$ for any $x \in Q$ and, therefore, we call $\epsilon_x(M)$ for $M \in At(\Phi)$ and $x \in Q$ \emph{local atoms} for $x$.  %Applied to our algebra of coherent sets of gambles this gives the following.
%\begin{theorem}
%Let $\Phi$ be the atomistic information algebra of coherent sets of gambles. Then for all $x \in Q$, the subalgebra $\epsilon_x(\Phi)$ is also atomistic and $At(\epsilon_x(\Phi)) = \epsilon_x(At(\Phi)) = \{\epsilon_x(M):M \in At(\Phi)\}$.
%\end{theorem}
%Indeed, they represent maximally informative pieces of information for question $x$ (or partition $P_x$).
%Here we ended the discussion about atoms in $\Phi$ made in \cite{kohlas21b}.
Local atoms $M_x =\epsilon_x(M)$ for $x$ induce a partition $At_x$ of $At(\Phi)$ with blocks $At(M_x)$. If $M $ and $M'$ belong to the same block, we say that $M \equiv_x M'$. 

Let us indicate with $Part(At(\Phi))$ the set of these partitions. As for $Part(\pspace)$, we can introduce a partial order on $Part(At(\Phi))$ defined as: $At_x \le At_y$ if $M \equiv_y M'$ implies $M \equiv_x M'$, for every $M,M' \in At(\Phi)$. $Part(At(\Phi))$ forms a lattice under this order where, in particular, $At_x \vee At_y$ is the partition obtained as the non-empty intersections of blocks of $At_x$ with blocks of $At_y$ \cite{graetzer03}.
We claim moreover that these partitions of $At(\Phi)$ mirror the partitions $\mathcal{P}_x \in \mathcal{Q}$. 
Before stating this main result, we need the following lemma.
\begin{lemma}
Let us consider $M,M' \in At(\Phi)$ and $x \in Q$. Then
\begin{equation*}
    M \equiv_x M' \iff \epsilon_x(M) = \epsilon_x(M') \iff M,M' \in At(\epsilon_x(M)).
\end{equation*}
Hence, $At_x \le At_y$ if and only if $At(\epsilon_x(M)) \supseteq At(\epsilon_y(M))$ for every $M \in At(\Phi)$.
\end{lemma}
\begin{proof}
If $M \equiv_x M'$, there exists a local atom $M_x$ such that $M,M' \in At(M_x)$. Therefore, $M, M' \ge M_x$ and $\epsilon_x(M), \epsilon_x(M') \ge M_x$ \cite[Lemma~ 15, item 3]{kohlas21b}. However, $\epsilon_x(M), \epsilon_x(M')$ and $M_x$ are all local atoms, hence $\epsilon_x(M)=\epsilon_x(M')= M_x$. The converse is obvious.

For the second part, let us suppose $At(\epsilon_y(M)) \subseteq At(\epsilon_x(M))$ for every $M \in At(\Phi)$, and consider $M',M'' \in At(\Phi)$ such that $M' \equiv_y M''$. Then $M',M'' \in At(\epsilon_y(M'))$ and hence $M',M'' \in At(\epsilon_x(M'))$, which implies $M' \equiv_x M''$. Vice versa, consider $At_x \le At_y$ and $M' \in At(\epsilon_y(M))$ for some $M ,M'\in At(\Phi)$. Then, $M \equiv_y M'$, hence $M \equiv_x M'$ and so $M' \in At(\epsilon_x(M))$.
\end{proof}
%This result is important to show, in the next section, that $At(\Phi)$ forms again a domain-free information algebra. Further, it is fundamental to show that the information algebra $\Phi$ can be embedded in $At(\Phi)$.
%This fact will be helpful in the next section to %It should be clear that these facts could also be expressed in the labeled view of the information algebra.
Now we can state the main result of this section.
\begin{theorem} \label{th:AtomSeparoid}
The map $\mathcal{P}_x \mapsto At_x$, from the lattice of partitions $(Part(\pspace),\leq)$  of $\pspace$ to the lattice partitions  $(Part(At(\Phi)),\leq)$  of $At(\Phi)$, preserves order and join. Furthermore it preserves also conditional independence relations, that is, $\mathcal{P}_x \bot \mathcal{P}_y \vert \mathcal{P}_z$ implies $At_x \bot At_y \vert At_z$.
\end{theorem}

\begin{proof}
%First of all, notice that  obvious.
%So, $M \equiv_x M'$ if and only if $M,M' \in At(\epsilon_x(M))=At(\epsilon_x(M'))$. Moreover, $At_x \le At_y$ if and only if $At(\epsilon_x(M)) \supseteq At(\epsilon_y(M))$ for every $M \in At(\Phi)$. 
%Indeed, let us consider $M \equiv_y M'$, then $M,M' \in At(\epsilon_y(M))$. However, $At(\epsilon_y(M)) \subseteq At(\epsilon_x(M))$ by hypothesis, hence $M,M' \in At(\epsilon_x(M))$, therefore $M \equiv_x M'$. Vice versa, consider $M' \in At(\epsilon_y(M))$ for some $M \in At(\Phi)$. Then, $M,M' \ge \epsilon_y(M)$, hence $M \equiv_y M'$. Therefore, $M \equiv_x M'$, hence $M' \in At(\epsilon_x(M))$.
%It is possible to notice that $(Part(At(\Phi)), \le)$ . Therefore, we can define an information order on $Part(At(\Phi))$ as: $At_x \le At_y$, if and only if $At(\epsilon_y(M)) \subseteq At(\epsilon_x(M))$, for any $M \in At(\Phi)$. Moreover, $(Part(At(\Phi)), \le)$ is a lattice, where $At_x \vee At_y = At_x \cap At_y$.
%As before we write $x \leq y$ for $P_x \leq P_y$ and $x \vee y$ for $P_x \vee P_y$.
If $x \leq y$, then $\epsilon_x(M) \leq \epsilon_y(M)$ for any atom $M \in At(\Phi)$ \cite[Lemma~ 15, item 4]{kohlas21b}. Therefore, $At(\epsilon_x(M)) \supseteq At(\epsilon_y(M))$ for any $M \in At(\Phi)$, hence $At_x \leq At_y$.
The converse is also true: indeed, if $At_x \leq At_y$,
then  $At(\epsilon_x(M)) \supseteq At(\epsilon_y(M))$ for any $M  \in At(\Phi)$. This implies in particular that $\epsilon_x(M) = \cap At(\epsilon_x(M))  \subseteq \cap At(\epsilon_y(M)) = \epsilon_y(M) $ for any $M  \in At(\Phi)$, thanks to the fact that $At(\Phi)$ is atomistic. Now, for any $D$ coherent, consider the family $\{M_j\}_{j \in J} \coloneqq At(D)$. Then we have:
\begin{equation*}
    \epsilon_x(\desirs) = \epsilon_x(\cap_{j \in J}  M_j) =  \cap_{j \in J} \epsilon_x(M_j) \subseteq \cap_{j \in J}\epsilon_y(M_j) = \epsilon_y( \cap_{j \in J}  M_j) = \epsilon_y(\desirs),
\end{equation*}
thanks to \cite[Theorem~ 17]{kohlas21b}. Therefore we have $\epsilon_x(\desirs) \subseteq \epsilon_y(\desirs)$ also for any $\desirs \in C$. Applying it to $\desirs \coloneqq \edesirs(\{f\})$ for every $f \in \gambles_x \setminus (\gambles^+_x \cup \{ f \in \gambles_x: f \le 0 \})$, we obtain that $\gambles_x \subseteq \gambles_y$, from which it follows that $x \leq y$ \cite[Section~ 3]{kohlas21b}. So the map $\mathcal{P}_x \mapsto At_x$ is an order isomorphism \cite[Def.~1.34]{daveypriestley97}, therefore it also preserves joins \cite[Prop.~2.19]{daveypriestley97}. 

%Now, %it is true that $At(M_{x \vee y}) = At(M_x) \vee  At(M_y)= At(M_x) \cap At(M_y)$ where $M_{x \vee y} = \epsilon_{x \vee y} (M), M_x= \epsilon_x(M), M_y = \epsilon_y(M)$ for $M \in At(\Phi)$. Indeed, 
%consider $M' \in At(\epsilon_{x \vee y}(M))$ for some $M \in At(\Phi)$. Then, $M' \ge \epsilon_{x \vee y} (M) \ge \epsilon_x(M), \epsilon_y(M)$ \cite[Lemma~ 15]{kohlas21b}, therefore $ M' \in At(\epsilon_x(M)) \cap At(\epsilon_y(M))$ and hence $At_{x \vee y} \le At_x \vee At_y$.

%Vice versa, let us consider $M'' \in At(\epsilon_x(M)) \cap At(\epsilon_y(M))$. 
%Then...

%Indeed, both $\epsilon_{x \vee y}(M'')$ and $M_{x \vee y}$ are local atoms so there are two possibilities: $\epsilon_{x \vee y}(M'') \cdot M_{x \vee y}=0$ or  $\epsilon_{x \vee y}(M'') =M_{x \vee y}$. If $\epsilon_{x \vee y}(M'') \cdot M_{x \vee y}=0$, then $\epsilon_{x \vee y}(M \cdot \epsilon_{x \vee y}(M'')) =0$ that means $M \cdot \epsilon_{x \vee y}(M'')=0$.***why is it not possible?*** 
%Therefore, $At_{x \vee y} =  At_x \vee At_y$.
%Consider two local atoms $M_x = \epsilon_x(M)$ and $M_y = \epsilon_y(M)$ so that $At(M_x) \cap At(M_y)$ is not empty, since it contains $M$. Then, we have $M_{x \vee y} = \epsilon_{x \vee y}(M)$, hence $M \in At(M_{x \vee y})$. It follows that $At(M_{x \vee y}) = At(M_x) \cap At(M_y)$, which shows that $A_{x \vee y} = At_x \vee At_y$, again in information order.

For the second part, %we write again $x \bot y \vert z$ for $P_x \bot P_y \vert z$. R
recall that $x \bot y \vert z$ if and only if $x \vee z \bot y \vee z \vert z$. Consider then local atoms $M_{x \vee z},M_{y \vee z}$ and $M_z$ so that
\begin{equation*} 
At(M_{x \vee z}) \cap At(M_z) \not= \emptyset, \quad At(M_{y \vee z}) \cap At(M_z) \not= \emptyset.
\end{equation*}
Hence, there is an atom $M' \in At(M_{x \vee z}) \cap At(M_z)$ and an atom $M''\in At(M_{y \vee z}) \cap At(M_z)$. Therefore, $M_{x \vee z} = \epsilon_{x \vee z}(M')$, $M_{y \vee z} = \epsilon_{y \vee z}(M'')$ and $M_z= \epsilon_z(M')=\epsilon_z(M'')$.
%so that $M_{x \vee z},M_z \leq M'$, hence $M_{x \vee z} \leq \epsilon_{x \vee z}(M')$. Since both $M_{x \vee z}$ and $\epsilon_{x \vee z}(M')$ are local atoms we conclude that $M_{x \vee z} = \epsilon_{x \vee z}(M')$. From $M_z \leq M'$ we obtain in the same way that $M_z = \epsilon_z(M')$. Similarly there is an atom $M'' \in At(M_{y \vee z}) \cap At(M_z)$ and as before we derive $M_{y \vee z} = \epsilon_{y \vee z}(M'')$ and $M_z = \epsilon_z(M'')$. 
Now, thanks to the Existential Quantifier axiom, we have:
\begin{equation*}
\epsilon_z(M_{x \vee z} \cdot M_{y \vee z} \cdot M_z) = \epsilon_z(M_{x \vee z} \cdot M_{y \vee z}) \cdot M_z.
\end{equation*}
Thanks to \cite[Theorem~ 16]{kohlas21b} and \cite[Lemma~ 15, item 6]{kohlas21b},\footnote{\cite[Theorem~ 16, item 2]{kohlas21b} indeed, can be rewritten also as follows: let $\desirs_1, \desirs_2 \in \Phi$ and $x,y,z \in Q$, if  $\desirs_1$ has support $x \vee z$, $\desirs_2$ has support $y \vee z$  and $x \bot y \vert z$ ,then $\epsilon_z(\desirs_1 \cdot \desirs_2) = \epsilon_z(\desirs_1) \cdot \epsilon_z(\desirs_2)$.} we obtain
\begin{align*}
\epsilon_z(M_{x \vee z} \cdot M_{y \vee z}) \cdot M_z  = \epsilon_z(\epsilon_{x \vee z}(M')) \cdot \epsilon_z(\epsilon_{y \vee z}(M'')) \cdot M_z  = \epsilon_z(M') \cdot \epsilon_z(M'') \cdot M_z \neq 0.
\end{align*}
Therefore $M_{x \vee z} \cdot M_{y \vee z} \cdot M_z \neq 0$ \cite[Lemma~ 15, item 2]{kohlas21b} and hence, since the algebra is atomic, there is an atom $M''' \in At(M_{x \vee z} \cdot M_{y \vee z} \cdot M_z)$. Then $M_{x \vee z}, M_{y \vee z}, M_z \le M'''$, whence $M''' \in At(M_{x \vee z}) \cap At(M_{y \vee z}) \cap At(M_z)$ and so
%From $x \vee z \bot y \vee z \vert z$ we obtain \cite[Theorem~ 16]{kohlas21b}:
%\begin{eqnarray*}
%\epsilon_z(M' \cdot M'') = \epsilon_z(M') \cdot \epsilon_z(M'') = M_z \not= 0.
%\end{eqnarray*}
%Therefore we have also $M' \cdot M'' \not= 0$. As the algebra is atomic, there is an atom $M''' \in At(M' \cdot M'')$ so that $M' \cdot M'' \leq M'''$. It follows that $M''' \geq M' \geq \epsilon_{x \vee z}(M') = M_{x \vee z}$, hence $M''' \in At(M_{x \vee z})$. In the same way, from $M''' \geq M''$ we obtain $M''' \in At(M_{y \vee z})$. Finally, $M''' \geq M' \geq \epsilon_z(M') = M_z$ implies $M''' \in At(M_z)$. So we see that $M''' \in At(M_{x \vee z}) \cap At(M_{y \vee z}) \cap At(M_z)$, hence
%\begin{eqnarray*}
%At(M_{x \vee z}) \cap At(M_{y \vee z}) \cap At(M_z) \not= \emptyset.
%\end{eqnarray*}
 $At_x \bot At_y \vert At_z$.
\end{proof}
%This proof shows in particular that
%$M \equiv_x M'$ if and only if $\epsilon_x(M') = \epsilon_x(M)$.
%We shall see in the next Section \ref{sec:SetAlg} that this is the base for an alternative view on the information algebra of coherent sets of gambles. 

%%%%%%%%%%%%%%%%%%%%%%%%%%%%%%%%%%%%%%%%%%%%%%%%%%%%%%%%%%%%
\section{Information Algebras Homomorphisms}\label{sec:homomorphism}
We are interested in homomorphisms between algebras:
%Given the fact that in the whole paper we concentrate ourselves only on domain-free information algebras, we give the definition of information algebras homomorphism only in this case.
\begin{definition}[Domain-free information algebras homomorphism]\label{def:homomorphism}
Let $(\Psi,\mathcal{Q}, \le_{\Psi},\bot_{\Psi},  \cdot_\Psi, E)$ and $(\Psi',\mathcal{Q}', \le_{\Psi'}, \bot_{\Psi'}, \cdot_{\Psi'}, E')$ be two domain-free information algebras, where $E \coloneqq \{ \epsilon_{\mathcal{P}}, \; \mathcal{P} \in \mathcal{Q}\}$ and $E' \coloneqq \{ \epsilon'_{\mathcal{P}'}, \; \mathcal{P}' \in \mathcal{Q}'\}$ are respectively the families of the extraction operators of the two algebras.
%, defined respectively from the sets of partitions $\mathcal{Q}$ and $\mathcal{Q}'$. 
%of partitions $\mathcal{Q}$, $\mathcal{Q}'$. 
A tuple $(f,h,g)$ of maps $f: \Psi \rightarrow \Psi'$, $h: \mathcal{Q} \rightarrow \mathcal{Q}'$ and $g:E \mapsto E'$ defined as $g: \epsilon_{\mathcal{P}} \rightarrow \epsilon'_{h(\mathcal{P})}$, is an homomorphism between $(\Psi,\mathcal{Q}, \le_{\Psi},\bot_{\Psi},  \cdot_\Psi, E)$ and $(\Psi',\mathcal{Q}', \le_{\Psi'}, \bot_{\Psi'}, \cdot_{\Psi'}, E')$ if and only if:
\begin{enumerate}
    \item $f(\psi \cdot_{\Psi} \phi) = f(\psi) \cdot_{\Psi'} f(\phi)$, for every $\phi, \psi \in \Psi$;
    \item $f(0_{\Psi}) = 0_{\Psi'}$ and $f(1_{\Psi})=1_{\Psi'}$, if we indicate with $0_{\Psi},1_{\Psi} $ and $0_{\Psi'},1_{\Psi'} $ respectively, the $0$ and the $1$ elements of $\Psi$ and $\Psi'$;
    \item if $\mathcal{P}_1 \le_{\Psi} \mathcal{P}_2$ then $h(\mathcal{P}_1) \le_{\Psi'} h(\mathcal{P}_2)$, for every $\mathcal{P}_1,\mathcal{P}_2 \in \mathcal{Q}$;
    \item $h(\mathcal{P}_1 \vee_{\Psi} \mathcal{P}_2) = h(\mathcal{P}_1) \vee_{\Psi'} h(\mathcal{P}_2) $ for every $\mathcal{P}_1,\mathcal{P}_2 \in \mathcal{Q}$, if we indicate with $\mathcal{P}_1 \vee_{\Psi} \mathcal{P}_2$, the join of $\mathcal{P}_1,\mathcal{P}_2$ with respect to $\le_{\Psi}$ and with $ h(\mathcal{P}_1) \vee_{\Psi'} h(\mathcal{P}_2)$, the join of $h(\mathcal{P}_1),h(\mathcal{P}_2)$ with respect to $\le_{\Psi'}$;
    \item $\mathcal{P}_1 \bot_{\Psi} \mathcal{P}_2 \vert \mathcal{P}$ implies $h(\mathcal{P}_1) \bot_{\Psi'} h(\mathcal{P}_2) \vert h(\mathcal{P})$ for every $\mathcal{P}_1,\mathcal{P}_2,\mathcal{P} \in \mathcal{Q}$;
    \item $f(\epsilon_{\mathcal{P}}(\psi))= g(\epsilon_\mathcal{P}) (f(\psi))$, for all $\psi \in \Psi$ and $\epsilon_\mathcal{P}\in E$ with $\mathcal{P} \in \mathcal{Q}$.
\end{enumerate}
\end{definition}
 If the maps are one-to-one, then $(\Psi,\mathcal{Q}, \le_{\Psi},\bot_{\Psi},  \cdot_\Psi, E)$ is said to be \emph{embedded} into $(\Psi',\mathcal{Q}', \le_{\Psi'}, \bot_{\Psi'}, \cdot_{\Psi'}, E')$.
If they are also bijiective, the homomorphism is said to be an \emph{isomorphism} between the two algebras. %$(\Psi,\mathcal{Q}, \le_{\Psi},\bot_{\Psi},  \cdot_\Psi, E)$ and $(\Psi',\mathcal{Q}', \le_{\Psi'}, \bot_{\Psi'}, \cdot_{\Psi'}, E')$.

This definition is an extension of the information algebras homomorphism given in \cite{kohlas17},\footnote{ In \cite{kohlas17} the set of questions $Q$ is used in place of the set of partitions $\mathcal{Q}$. Here we need to be more explicit about partitions.} for domain-free information algebras for which $\mathcal{Q}$ is potentially different from $\mathcal{Q}'$. If $\mathcal{Q}=\mathcal{Q}'$, or equivalently $Q=Q'$, it collapses to the simpler definition in \cite{kohlas17}.

\begin{comment}
In case of commutative partitions, the above definition simplifies.

\begin{definition}[Information algebras homomorphism - commutative partitions]
Let $(\Psi,Q, \le_{\Psi},\bot_{\Psi},  \cdot_\Psi, E)$ and $(\Psi',Q', \le_{\Psi'}, \bot_{\Psi'}, \cdot_{\Psi'}, E')$ be two domain-free information algebras where $Q$ and $Q'$ are constituted only by commutative partitions and  where $E \coloneqq \{ \epsilon_x, \; x \in Q\}$ and $E' \coloneqq \{ \epsilon'_{x'}, \; x' \in Q'\}$ are respectively the families of the extractors operators in the two algebras, defined starting respectively from the two sets of questions $Q$ and $Q'$.

A tuple $(f,g)$ of maps $f: \Psi \rightarrow \Psi'$, $g: E \rightarrow E'$, is an homomorphism between $(\Psi,Q, \le_{\Psi},\bot_{\Psi},  \cdot_\Psi, E)$ and $(\Psi',Q', \le_{\Psi'}, \bot_{\Psi'}, \cdot_{\Psi'}, E')$ if and only if:
\begin{enumerate}
    \item $f(\psi \cdot_{\Psi} \phi) = f(\psi) \cdot_{\Psi'} f(\phi)$, for every $\phi, \psi \in \Psi$;
    \item $f(0_{\Psi}) = 0_{\Psi'}$ and $f(1_{\Psi})=1_{\Psi'}$, if we indicate with $0_{\Psi},1_{\Psi} $ and $0_{\Psi'},1_{\Psi'} $ respectively, the $0$ and the $1$ elements of $\Psi$ and $\Psi'$;
    
    \item $g(\epsilon_x \circ \epsilon_) = g(\epsilon_x) \vee_{\Psi'} g(\epsilon_y) $  for every $\epsilon_x, \epsilon_y \in E$;
    \item $f(\epsilon_x(\psi))= g(\epsilon_x) (f(\psi))$, for all $\psi \in \Psi$ and $\epsilon_x \in E$.
\end{enumerate}
\end{definition}
\end{comment}
\section{Set Algebras} \label{sec:SetAlg}

Archetypes of information algebras are so-called set algebras, where the elements are subsets of some universe, combination is intersection, and extraction is related to so-called saturation operators. Starting with the set $\pspace$ of possibilities, representing possible worlds, pieces of information may be given by subsets $S$ of $\pspace$, meaning that the unknown world must be an element of $S$. As before, questions $x \in Q$ are modeled by partitions $\mathcal{P}_x$ or equivalence relation $\equiv_x$, where $\omega \equiv_x \omega'$ means that question $x$ has the same answer in possible worlds $\omega$ and $\omega'$. We first specify the set algebra of subsets of $\pspace$ and show then that this algebra may be embedded into the information algebra of coherent sets. Conversely, we show that the algebra $\Phi$ of coherent sets of gambles may itself be embedded into a set algebra of its atoms, so is, in some precise sense, itself a set algebra. This is a general result for atomistic information algebras \cite{kohlas03,kohlasschmid16}.

To any partition $\mathcal{P}_x$ of $\pspace$ there corresponds a saturation operator defined for any subset $S \subseteq \pspace$ by
\begin{equation}\label{eq:saturOp}
\sigma_x(S) \coloneqq \{\omega \in \pspace: (\exists \omega' \in S)\; \omega \equiv_x \omega'\}.
\end{equation}
The following are well-known properties of saturation operators.

\begin{lemma} \label{le:SatOps}
For all $S,T \subseteq \pspace$ and any partition $\mathcal{P}_x$ of $\pspace$:
\begin{enumerate}
\item $\sigma_x(\emptyset) = \emptyset$,
\item $S \subseteq \sigma_x(S)$,
\item $\sigma_x(\sigma_x(S) \cap T) = \sigma_x(S) \cap \sigma_x(T)$,
\item $\sigma_x(\sigma_x(S)) = \sigma_x(S)$,
\item $S \subseteq T \Rightarrow \sigma_x(S) \subseteq \sigma_x(T)$,
\item $\sigma_x(\sigma_x(S) \cap \sigma_x(T)) = \sigma_x(S) \cap \sigma_x(T)$.
%\item $\sigma_x(S \cup T) = \sigma_x(S) \cup \sigma_x(T)$.
\end{enumerate}
\end{lemma}

\begin{comment}
\begin{proof}
Items 1, 2 and 4 are obvious. For 3 note that $\sigma_x(T) \supseteq T$, hence $\sigma_x(\sigma_x(S) \cap T) \subseteq \sigma_x(S) \cap \sigma_x(T)$. So consider an element $\omega \in  \sigma_x(S) \cap \sigma_x(T)$, such that for some elements $\omega' \in S$ and $\omega'' \in T$ we have $\omega \equiv_x \omega'$ and $\omega \equiv_x \omega'' $. By transitivity it follows that $\omega'' \equiv_x \omega'$ so that $\omega'' \in \sigma_x(S)$. But then $\omega \equiv_x \omega'' \in \sigma_x(S) \cap T$ implies $\omega \in \sigma_x(\sigma_x(S) \cap T)$ and this proves item 3. Item 4 follows from 3: $\sigma_x(\sigma_x(S)) = \sigma_x(\sigma_x(S) \cap \pspace) = \sigma_x(S) \cap \sigma_x(\pspace) = \sigma_x(S)$. Then item 6 follows $\sigma_x(\sigma_x(S) \cap \sigma_x(T)) = \sigma_x(S) \cap \sigma_x(\sigma_x(T)) = \sigma_x(S) \cap \sigma_x(T)$. Finally, 7. is immediate.
\end{proof}
\end{comment}

\begin{proof}
Items 1, 2, 4, 5 are obvious. 
For item 6, consider $\omega \in \sigma_x(\sigma_x(S) \cap \sigma_x(T))$. Then there is a $\omega' \in \sigma_x(S) \cap \sigma_x(T)$ so that $\omega \equiv_x \omega'$. In particular, $\omega' \in \sigma_x(S)$, hence $\omega \in \sigma_x(\sigma_x(S))=\sigma_x(S)$ by item 4. At the same time, $\omega' \in \sigma_x(T)$, hence $\omega \in \sigma_x(\sigma_x(T))=\sigma_x(T)$. Then $\omega \in \sigma_x(S) \cap \sigma_x(T)$. By item 2 we must then have equality.
Regarding item 3, $\sigma_x(\sigma_x(S) \cap T) \subseteq \sigma_x(S) \cap \sigma_x(T)$ by item 2, 5 and 6. So consider an element $\omega \in  \sigma_x(S) \cap \sigma_x(T)$. Then, there exist $\omega' \in S$ and $\omega'' \in T$ such that $\omega \equiv_x \omega'$ and $\omega \equiv_x \omega'' $. By transitivity it follows that $\omega'' \equiv_x \omega'$ so that $\omega'' \in \sigma_x(S)$. But then $\omega \equiv_x \omega'' \in \sigma_x(S) \cap T$ implies $\omega \in \sigma_x(\sigma_x(S) \cap T)$ and this proves item 3. 
\end{proof}

%Note that the first three items of this theorem say that $\sigma_x$ is an existential quantifier relative to intersection as combination. Further, either we have the top partition whose blocks are single elements of $\pspace$ among the partitions $P_x$ or we consider only subsets $S$ which are saturated relative to a partition $P_x$ for $x \in Q$. In both cases any element of the algebra has a support $x$, $\sigma_x(S) = S$. In the first case this is for any subset $S$ the top partition. If $P_x \leq P_y$, then $\omega \equiv_y \omega'$ implies $\omega \equiv_x \omega'$, so that $\sigma_y(S) \subseteq \sigma_x(S)$. If $x$ is a support of $S$, we have $S \subseteq \sigma_y(S) \subseteq \sigma_x(S) = S$, hence $\sigma_y(S) = S$. So, the support axiom is satisfied. Also if $S$ and $T$ are two subsets of $\pspace$ with support $x$ and $y$ respectively, then they have also both supports $x \vee y$ and this shows also that $S \cap T$ is saturated relative to the partition $P_x \vee P_y = P_{x \vee y}$. So, the family of subsets saturated with respect to some partition $P_x$ for $x \in Q$ is closed under intersection and forms a commutative semigroup with the empty sets as null and $\pspace$ as unit elements and  $(\mathcal{P}(\pspace);\cap)$ is anyway a commutative semigroup. It remains only to verify the extraction property to conclude that the power set of $\pspace$ together with the partitions $P_x$ for $x \in Q$ form an information algebra.

Note that the first three items of this theorem imply that $\sigma_x$ is an existential quantifier relative to intersection as combination. 
This is a first step to construct a domain-free information algebra of subsets of $\pspace$.
Then we limit possible questions to the same join semilattice $(Q, \le)$ considered in the previous sections. Moreover, we consider on it the quasi-separoid three-place relation: $x \perp y \vert z$ with $x,y,z \in Q$, defined before.

Now, we want the support axiom to be satisfied. Hence, if $\partit_\top$ belongs to $Q$, then we have $\sigma_\top(S) = S$ for all $S \subseteq \pspace$. Otherwise, we must limit ourselves to the subsets of $\pspace$ for which there is a support $x \in Q$. We call these sets \emph{saturated} with respect to some $x \in Q$, and we indicate them with $P_Q(\pspace)$ or more simply with $P_Q$ when no ambiguity is possible. Clearly, if the top partition belongs to $Q$, $P_Q(\pspace) = P(\pspace)$, the power set of $\pspace$. So in what follows we can refer more generally to sets in $P_Q(\pspace)$. Note that in particular $\pspace, \emptyset \in P_Q(\pspace)$ for every join semilattice $(Q, \le)$.
At this point the support axiom is satisfied. Indeed, if $x \leq y$ with $x,y \in Q$, then $\omega \equiv_y \omega'$ implies $\omega \equiv_x \omega'$, so that $\sigma_y(S) \subseteq \sigma_x(S)$. Then, if $x$ is a support of $S$, we have $S \subseteq \sigma_y(S) \subseteq \sigma_x(S) = S$, hence $\sigma_y(S) = S$.
Moreover $(P_Q(\pspace),\cap)$ is a commutative semigroup with the empty set as the null element and $\pspace$ as the unit. 
Indeed, the only property we need to prove, is that $P_Q(\pspace)$ is closed under intersection. Then, let us consider $S$ and $T$, two subsets of $\pspace$ with support $x \in Q$ and $y \in Q$ respectively. Then they have also both supports $x \vee y $ that belongs to $Q$, because it is a join semilattice. Therefore, thanks to Lemma \ref{le:SatOps}, we have
\begin{equation*}
\sigma_{x \vee y}(S \cap T) = \sigma_{x \vee y}(\sigma_{x \vee y}(S) \cap \sigma_{x \vee y}(T)) =  \sigma_{x \vee y}(S) \cap \sigma_{x \vee y}(T)= S \cap T.
\end{equation*}
So, $P_Q(\pspace)$ is closed under intersection.
It remains only to verify the extraction property to conclude that $P_Q(\pspace)$ %together with partitions $\partit_x$ for $x \in Q$ and a family of saturation operators $\sigma_x$ with $x \in Q$, 
forms a domain-free information algebra.

\begin{theorem} \label{th:ExtrPropSets}
Given $x,y,z \in Q$, suppose $x \vee z \bot y \vee z\vert z$. Then, for any $S \in P_Q(\pspace)$,
\begin{equation*}
\sigma_{y \vee z}(\sigma_x(S)) = \sigma_{y \vee z}(\sigma_z(\sigma_x(S))).
\end{equation*}
\end{theorem}

\begin{proof}
From $\sigma_z(\sigma_x(S)) \supseteq \sigma_x(S)$ we obtain $\sigma_{y \vee z}(\sigma_z(\sigma_x(S)) \supseteq \sigma_{y \vee z}(\sigma_x(S))$. Consider therefore an element $\omega \in \sigma_{y \vee z}(\sigma_z(\sigma_x(S)))$. Then there are elements $\mu,\mu'$ and $\omega'$ so that $\omega \equiv_{y \vee z} \mu \equiv_z \mu' \equiv_x \omega'$ and $\omega' \in S$. This means that $\omega,\mu$ belong to some block $B_{y \vee z}$ of partition $\mathcal{P}_{y \vee z}$, $\mu,\mu'$ to some block $B_z$ of partition $\mathcal{P}_z$ and $\mu',\omega'$ to some block $B_x$ of partition $\mathcal{P}_x$. It follows that $B_x \cap B_z \not= \emptyset$ and $B_{y \vee z} \cap B_z \not= \emptyset$. Then $x \vee z \bot y \vee z \vert z$ implies, thanks to properties of a separoid, that $x \bot y \vee z \vert z$. Therefore, we have $B_x \cap B_{y \vee z} \cap B_z \not= \emptyset$, and in particular, $B_x \cap B_{y \vee z} \not= \emptyset$. So there is a $\lambda \in B_x \cap B_{y \vee z}$ such that
$\omega \equiv_{y \vee z} \lambda \equiv_x \omega' \in S$, hence $\omega \in \sigma_{y \vee z}(\sigma_x(S))$. So we have $\sigma_{y \vee z}(\sigma_x(S)) = \sigma_{y \vee z}(\sigma_z(\sigma_x(S)))$.
\end{proof}

%Recall that $x \bot y \vert z$ is equivalent to $x \vee z \bot y \vee z \vert z$, so that according to this theorem saturation operators satisfy the Extraction axiom.
Hence, these algebras of sets, with intersection as combination and saturation as extraction, form domain-free information algebras. Such algebras will be called \textit{set algebras}.
A set algebra of subsets of $\pspace$ can be embedded in the information algebra of coherent sets of gambles defined on $\pspace$. For any set $S \in P_Q(\pspace)$, define
\begin{equation*}
\desirs_S \coloneqq \{f \in \mathcal{L}(\pspace): \inf_{\omega \in S} f(\omega) >0 \} \cup \mathcal{L}^+(\pspace). %\edesirs(\gambles_S^+ ) = \{f \in \mathcal{L}(\pspace):f|S \ge 0, \; f|S \neq 0 \} \cup \mathcal{L}(\pspace)^+.
\end{equation*}
If $S \neq \emptyset$, this is clearly a coherent set. %even a strictly desirable set of gambles. 
The next theorem shows that the map $S \mapsto \desirs_S$ together with the map $\sigma_x \mapsto \epsilon_x$ is an information algebra homomorphism, according to the simpler definition given in \cite{kohlas17}. It can be applied in fact, because in this case the set of partitions/questions analyzed by the two information algebras is the same.

\begin{theorem} \label{th:InfAlgHom}
Let $S,T \in P_Q(\pspace)$ and $x \in Q$. Then
\begin{enumerate}
\item $\desirs_S \cdot \desirs_T = \desirs_{S \cap T}$,
\item $\desirs_\emptyset = \mathcal{L}(\pspace)$, $\desirs_\pspace = \mathcal{L}^+(\pspace)$,
\item $\epsilon_x(\desirs_S) = \desirs_{\sigma_x(S)}$.
\end{enumerate}

\end{theorem}

\begin{proof}
1. Note that $D_S= \gambles^+$ or $D_T= \gambles^+$ if and only if $S=\pspace$ or $T= \pspace$. Clearly in this case we have immediately the result.
The same is true if $D_S= \gambles$ or $D_T= \gambles$, which is equivalent to have $S= \emptyset $ or $T= \emptyset$. Now suppose $D_S,D_T \neq \gambles^+$ and
$D_S,D_T \neq \gambles$.
If $S \cap T = \emptyset$, then $\desirs_{S \cap T} = \mathcal{L}(\pspace)$. %And  by definition $\desirs_S \cdot \desirs_T = \mathcal{C}(\desirs_S \cup \desirs_T)$. 
Consider $f \in \desirs_S$ and $g \in \desirs_T$. Since $S$ and $T$ are disjoint, we have $\tilde{f} \in \desirs_S$ and $\tilde{g} \in \desirs_T$, where $\tilde{f}, \; \tilde{g}$ are defined in the following way:
 \begin{align*}
\tilde{f}(\omega) \coloneqq \left\{ \begin{array}{ll} f(\omega) & \textrm{for}\ \omega \in S, \\ -g(\omega) & \textrm{for}\ \omega \in T, \\ 0 & \textrm{for}\ \omega \in  (S\cup T)^c,\end{array} \right. &&
\tilde{g}(\omega) \coloneqq \left\{ \begin{array}{ll} -f(\omega) & \textrm{for}\ \omega \in S , \\ g(\omega) & \textrm{for}\ \omega \in T, \\ 0 & \textrm{for}\ \omega \in (S\cup T)^c. \end{array} \right.
\end{align*}
However, $\tilde{f} + \tilde{g} = 0 \in \mathcal{E}(\desirs_S \cup \desirs_T)$, hence $\desirs_S \cdot \desirs_T = \mathcal{L}(\pspace) = \desirs_{S \cap T}$.
Assume then that $S \cap T \not= \emptyset$. Note that $\desirs_S \cup \desirs_T \subseteq \desirs_{S \cap T}$ so that $\desirs_S \cdot \desirs_T$ is coherent and $\desirs_S \cdot \desirs_T \subseteq \desirs_{S \cap T}$. Consider then a gamble $f \in \desirs_{S \cap T}$. Select a $\delta > 0$ and define two functions
\begin{align*}
f_1(\omega) \coloneqq \left\{ \begin{array}{ll} 1/2f(\omega) & \textrm{for}\ \omega \in (S \cap T), \\ \delta & \textrm{for}\ \omega \in S \setminus T, \\ f(\omega) - \delta & \textrm{for}\ \omega \in T \setminus S,
\\ 1/2f(\omega) & \textrm{for}\ \omega \in (S \cup T)^c, \end{array} \right.
&&
f_2(\omega) \coloneqq \left\{ \begin{array}{ll} 1/2f(\omega) & \textrm{for}\ \omega \in (S \cap T), \\ f(\omega) - \delta & \textrm{for}\ \omega \in S \setminus T,
\\ \delta & \textrm{for}\ \omega \in T \setminus S, \\
1/2f(\omega) & \textrm{for}\ \omega \in (S \cup T)^c . \end{array} \right.
\end{align*}
Then $f = f_1 + f_2$ and $f_1 \in \desirs_S$, $f_2 \in \desirs_T$. Therefore $f \in \edesirs(\desirs_S \cup \desirs_T)= \mathcal{C}(\desirs_S \cup \desirs_T) \eqqcolon  \desirs_S \cdot \desirs_T$, hence $\desirs_S \cdot \desirs_T = \desirs_{S \cap T}$. 

2. Both have been noted above. 

3. First of all it can be noticed that, if $S \in P_Q(\pspace)$, then $\sigma_x(S) \in P_Q(\pspace)$. So $\desirs_{\sigma_x(S)}$ is well defined. Furthermore, if $S$ is empty, then $\epsilon_x(\desirs_\emptyset) = \mathcal{L}(\pspace)$ so that item 3 holds in this case. Hence, assume $S \not= \emptyset$. Then we have
\begin{equation*}
\epsilon_x(\desirs_S) \coloneqq \mathcal{C}(\desirs_S \cap \mathcal{L}_x) = \posi(\mathcal{L}^+(\pspace) \cup (\desirs_S \cap \mathcal{L}_x)). 
\end{equation*}
Consider a gamble $f \in \desirs_S \cap \mathcal{L}_x$. We have $\inf_S f > 0$ and $f$ is $x$-measurable. If $\omega \equiv_x \omega'$ for some $\omega' \in S$ and $\omega \in \pspace$, then $f(\omega) = f(\omega')$. Therefore $\inf_{\sigma_x(S)} f = \inf_S f > 0$, hence $f \in D_{\sigma_x(S)}$. Then $\mathcal{C}(\desirs_S \cap \mathcal{L}_x) \subseteq \mathcal{C}(\desirs_{\sigma_x(S)})=\desirs_{\sigma_x(S)}$.
Conversely, consider a gamble $f \in \desirs_{\sigma_x(S)}$. $\desirs_{\sigma_x(S)}$ is a strictly desirable set of gambles.\footnote{$\lpr(f) \coloneqq \inf_S(f)$ for every $f \in \gambles$ with $S \neq \emptyset$, is a coherent lower prevision \cite{walley91}.} Hence, if $f \in \desirs_{\sigma_x(S)}$, $f \in \mathcal{L}^+(\pspace)$ or there is $\delta >0$ such that $f-\delta \in \desirs_{\sigma_x(S)}$. If $f \in \mathcal{L}^+(\pspace)$, then $f \in \epsilon_x(\desirs_S)$. Otherwise, let us define for every $\omega \in \pspace$, $g(\omega) \coloneqq \inf_{\omega' \equiv_x  \omega} f(\omega') - \delta$.
If $\omega \in S$, then $g(\omega) > 0$ since $\inf_{\sigma_x(S)} (f- \delta) > 0$. So we have $\inf_S g \ge 0$ and $g$ is $x$-measurable. However, $ \inf_S ( g + \delta) = \inf_S g + \delta >0$ hence $(g + \delta) \in \desirs_S \cap \mathcal{L}_x$ and $f \geq g+ \delta$. Therefore %***original: $g \in \mathcal{C}(D_A \cap \mathcal{L}_S)$***new version 
$f \in \mathcal{C}(\desirs_S \cap \mathcal{L}_x)$.

\end{proof}

%Let us call $\tilde{\Phi} \coloneqq \{ D_S : S \in Q\}$. You can notice that, $(\tilde{\Phi}, Q, \le, \bot \cdot, \{\epsilon_x\}_{x \in Q}\}$ is an information algebra.
%Indeed, clearly $\tilde{\Phi}$ satisfies axioms 1,2,3,4 that characterize a domain-free information algebra. Regarding the support axiom, you can notice that if $S \in P_Q(\pspace)$, $D_S$ has support in $x$. So it is satisfied.

%This theorem therefore, shows that $(f,g,h)$ where $f: P_Q(\pspace) \mapsto \tilde{\Phi}$, $g:Q \mapsto Q$, $h: \{\epsilon_x\}_{x \in Q} \mapsto  \{\sigma_x\}_{x \in Q} $ are defined as $ f: S \rightarrow \desirs_S$, $g: x \rightarrow x$ and $h: \epsilon_x \rightarrow \sigma_x$, is an information algebra homomorphism (or, following the simpler definition in \cite{kohlas17}, $f$ is an homomorphism). Since $f$,$g$ and $h$ are also obviously one-to-one (or simply $f$ is one-to-one), it is in particular an embedding of the set algebra $P_Q(\pspace)$ into $\tilde{\Phi}$. Further $\tilde{\Phi}$ is a subalgebra of $\Phi$, therefore it is also an embedding of $P_Q(\pspace)$ into the information algebra of coherent sets.
Item 3. guarantees that, if $S \in P_Q(\pspace)$, then there exists an $x \in Q$ such that $x$ is a support of $D_S$. Notice moreover that the two maps are one-to-one, therefore it is in particular an embedding of the set algebra $P_Q(\pspace)$ into $\Phi(\pspace)$.

Next we construct a set algebra of subsets of $At(\Phi)$. For this purpose we consider %define for any $x \in Q$ equivalence relations $\equiv_x$ in $At(\Phi)$ by
%\begin{equation*}
%M \equiv_x M' \Leftrightarrow \epsilon_x(M) = \epsilon_x(M').
%\end{equation*}
%We denote the associated partitions of $At(\Phi)$ by $Q_x$ and the related saturation operator by $\sigma_x$. By Theorem \ref{th:AtomSeparoid} and Theorem \ref{th:ExtrPropSets} applied to subsets of $At(\Phi)$ we see that the subsets of atoms of $\Phi$ form an information algebra, that is, a set algebra with intersection as combination and saturation relative to partitions $Q_x$ as extraction.
%We noticed in Section \ref{sec:Atoms} that the associated partitions of $At(\Phi)$ coincide with $At_x$.  We denote 
the set of partitions $At_x$ with $x \in Q$. We denote them as $Part_Q(At(\Phi))$. Moreover, we indicate with $\sigma_x$ the related saturation operators defined similarly to \eqref{eq:saturOp}, and with $At_Q(\Phi)$ the subsets of $At(\Phi)$ saturated with respect to some $At_x \in Part_Q(At(\Phi))$. %We denote moreover the related saturation operator by . 
By %Theorem \ref{th:AtomSeparoid} and 
Theorem \ref{th:AtomSeparoid} restricted to $\mathcal{P}_x$ with $x \in Q$, it is possible to derive that, if $(Q, \le) $ is a join semilattice, then $(Part_Q(At(\Phi)), \le)$ is also a join semilattice with $At_x \bot At_y \vert At_z$ a quasi-separoid \cite[Theorem 2.6]{kohlas17}.
So, thanks to Lemma \ref{le:SatOps} and the reasoning below, also $At_Q(\Phi)$ is a set algebra with intersection as combination and saturation relative to partitions $At_x$ as extraction. Moreover, thanks again to Theorem \ref{th:AtomSeparoid}, we know that $h: \mathcal{P}_x \mapsto At_x$, satisfies items 3,~4 and~5 of Definition~\ref{def:homomorphism}. Therefore, we need only an analog of Theorem~\ref{th:InfAlgHom} for $f: D \mapsto At(D)$ and $g: \epsilon_x \mapsto \sigma_x$, to conclude that $(f,h,g)$ is an information algebra homomorphism between $\Phi$ and $At_Q(\Phi)$.

\begin{theorem} \label{th:EmbedInSetAlg}
For any element $\desirs_1,\desirs_2$ and $\desirs$ of $\Phi$ and all $x \in Q$,
\begin{enumerate}
\item $At(\desirs_1 \cdot \desirs_2) = At(\desirs_1) \cap At(\desirs_2)$,
\item $At(\mathcal{L}(\pspace)) = \emptyset$, $At(\mathcal{L}(\pspace)^+) = At(\Phi)$,
\item $At(\epsilon_x(\desirs)) = \sigma_x(At(\desirs))$.
\end{enumerate}
\end{theorem}
 
\begin{proof}
Item 2 is obvious. 
%If $\desirs_1 \cdot \desirs_2 = 0$, then $At(\desirs_1 \cdot \desirs_2) = \emptyset$. But then $At(\desirs_1) \cap At(\desirs_2) = \emptyset$ too, since otherwise there is an atom $M \in At(\desirs_1) \cap At(\desirs_2)$ and $\desirs_1,\desirs_2 \leq M$ implies $\desirs_1 \cdot \desirs_2 \leq M$ and $M \in At(\desirs_1 \cdot \desirs_2) \not= \emptyset$ against the assumption. And, vice versa, if $At(\desirs_1) \cap At(\desirs_2) = \emptyset$, then $At(\desirs_1 \cdot \desirs_2) = \emptyset$ too because otherwise there would by an atom $M \geq \desirs_1 \cdot \desirs_2 \geq \desirs_1,\desirs_2$. 
If %$\desirs_1 \cdot \desirs_2 \not= 0$ then 
there is a an atom $M \in At(\desirs_1 \cdot \desirs_2)$, then $M \geq \desirs_1 \cdot \desirs_2 \geq \desirs_1,\desirs_2$ and thus  $M \in At(\desirs_1)$ and $M \in At(\desirs_2)$, hence $M \in At(\desirs_1) \cap At(\desirs_2)$. Conversely, if $M \in At(\desirs_1) \cap At(\desirs_2)$, then $\desirs_1,\desirs_2 \leq M$, hence $\desirs_1 \cdot \desirs_2 \leq M$ and $M \in At(\desirs_1 \cdot \desirs_2)$. This shows that $At(\desirs_1 \cdot \desirs_2) = At(\desirs_1) \cap At(\desirs_2)$.

Furthermore, if $\epsilon_x(\desirs) = 0$, then $\desirs = 0$ and $At(\desirs) = \emptyset$, hence $\sigma_x(\emptyset) = \emptyset$ and vice versa \cite[Lemma~ 15, item 2]{kohlas21b}.  Assume therefore $At(\desirs) \not= \emptyset$ and consider $M \in \sigma_x(At(\desirs))$. There is then a $M' \in At(\desirs)$ so that $\epsilon_x(M) = \epsilon_x(M')$. But $\desirs \leq M'$, hence $\epsilon_x(\desirs) \leq \epsilon_x(M') = \epsilon_x(M) \leq M$. Thus $M \in At(\epsilon_x(\desirs))$.
Conversely consider $M \in At(\epsilon_x(\desirs))$. We claim that $\epsilon_x(M) \cdot \desirs \not= 0$. Because otherwise $0 = \epsilon_x(\epsilon_x(M) \cdot \desirs) = \epsilon_x(M) \cdot \epsilon_x(\desirs) = \epsilon_x(M \cdot \epsilon_x(\desirs))$, which is not possible since $\epsilon_x(\desirs) \leq M$. So there is an $M' \in At(\epsilon_x(M) \cdot \desirs)$ so that $\desirs \leq \epsilon_x(M) \cdot \desirs \leq M'$. We conclude that $M' \in At(\desirs)$. Furthermore, $\epsilon_x(\epsilon_x(M) \cdot \desirs) = \epsilon_x(M) \cdot \epsilon_x(\desirs) \leq \epsilon_x(M')$. It follows that $\epsilon_x(M') \cdot \epsilon_x(M) \cdot \epsilon_x(\desirs) \not= 0$ and therefore $\epsilon_x(M) \cdot \epsilon_x(M') \not= 0$. But $\epsilon_x(M) \cdot \epsilon_x(M') = \epsilon_x(M \cdot \epsilon_x(M'))$ so that $M \cdot \epsilon_x(M') \not= 0$ and therefore $\epsilon_x(M') \leq M$ since $M$ is an atom, and then $\epsilon_x(M') \leq \epsilon_x(M)$. Proceed  in the same way from $\epsilon_x(M) \cdot \epsilon_x(M') = \epsilon_x(M' \cdot \epsilon_x(M))$ to obtain $\epsilon_x(M) \leq \epsilon_x(M')$. So finally $\epsilon_x(M) = \epsilon_x(M')$, which together with $M' \in At(\desirs)$ tells us that $M \in \sigma_x(At(\desirs))$. This means that $At(\epsilon_x(\desirs)) = \sigma_x(At(\desirs))$.
\end{proof}
Item 3. again guarantees that if $D \in \Phi$, then $At(D) \in At_Q(\Phi)$. Moreover, since $\Phi$ is atomistic, the maps $f,h,g$ are all one-to-one and then the homomorphism is an embedding. We can say therefore that $\Phi$ is in fact a set algebra. %This is an important representation theorem for the information algebra of coherent sets of gambles \cite{kohlas03,kohlasschmid21}. 

%The maps $f,h$ are in particular one-to-one, hence an embedding of $\Phi$ into the set algebra of subsets of $At(\Phi)$. Therefore $\Phi$ is in fact a set algebra. This is a representation theorem for the information algebra of coherent sets of gambles \cite{kohlas03,kohlasschmid21}. %The set algebra of subsets of $\pspace$ is via the maps $S \mapsto \desirs_S \mapsto At(\desirs_S)$ also embedded into the set algebra of subsets of $At(\Phi)$. But recall that the $\desirs_S$ are strictly desirable sets and belong to the subalgebra $\Phi^+$, form in fact a subalgebra of this subalgebra. We shall see that the algebra of strictly desirable gambles has their own representation as a set algebra of their own atoms, see Sections \ref{sec:LowUpPrev} and  \ref{sec:LinPrev} where these issues are clarified. 
\section{Conclusions}
This  paper  presents  an extension of our work on information  algebras  related  to gambles on a possibility set that is not necessarily multivariate \cite{kohlas21b}. In particular, here we analyze the relation between the domain-free version of the information algebra of coherent sets of gambles and the archetypes of information algebras, i.e., sets algebras. Specifically, we show that it is in fact a set algebra. These facts could also be expressed equivalently in the \emph{labeled view} of information algebras, better adapted to computational purposes \cite{kohlas03,kohlas21b}. This is left for future work, along with other aspects such as the question of conditioning.

%There are two different versions of information algebras: a \emph{domain-free} one, better suited for theoretical studies, and a \emph{labeled one}, better adapted to computational purposes (see \cite{kohlas03}). They are closely related and each one can be derived or reconstructed from the other. %The domain-free version is better suited for theoretical studies, since it is a structure of universal algebra, whereas the labeled one is better adapted to computational purposes (see \cite{kohlas03}).
 % In this paper we treat only the domain-free case. However, it must be intended that all the results found here could also been expressed in the labeled view.

\begin{comment}

For citations of references, we prefer the use of square brackets
and consecutive numbers. Citations using labels or the author/year
convention are also acceptable. The following bibliography provides
a sample reference list with entries for journal
articles~\cite{ref_article1}, an LNCS chapter~\cite{ref_lncs1}, a
book~\cite{ref_book1}, proceedings without editors~\cite{ref_proc1},
and a homepage~\cite{ref_url1}. Multiple citations are grouped
\cite{ref_article1,ref_lncs1,ref_book1},
\cite{ref_article1,ref_book1,ref_proc1,ref_url1}.
\end{comment}
\bibliographystyle{splncs04}
%\bibliography{mybibliography}
%
%ORIGINAL

%%%%%%%%%%%%%%%%%%%%%%%%%%%%%%%%%%%%%%%%%%%%%%%%%%%%%%%%%%%%
\end{document}